\documentclass{amsart}
\usepackage{hyperref}
\usepackage{enumerate}
\newcommand{\arxiv}[1]{\href{http://arxiv.org/#1}{arXiv:#1}}
\newcommand*{\mailto}[1]{\href{mailto:#1}{\nolinkurl{#1}}}

\newtheorem{theorem}{Theorem}[section]
\newtheorem{lemma}[theorem]{Lemma}
\newtheorem{corollary}[theorem]{Corollary}

\numberwithin{equation}{section}

\newcommand{\C}{\mathbb{C}}
\newcommand{\R}{\mathbb{R}}
\newcommand{\N}{\mathbb{N}}
\newcommand{\E}{\mathrm{e}}
\newcommand{\I}{\mathrm{i}}

\newcommand{\siu}{\sigma^{\mathrm{u}}}
\newcommand{\sil}{\sigma^{\mathrm{l}}}
\newcommand{\sipmu}{\sigma_\pm^{\mathrm{u}}}
\newcommand{\sipml}{\sigma_\pm^{\mathrm{l}}}
\newcommand{\sipmul}{\sigma_\pm^{\mathrm{u,l}}}
\newcommand{\simpul}{\sigma_\mp^{\mathrm{u,l}}}
\newcommand{\lau}{\la^{\mathrm{u}}}
\newcommand{\lal}{\la^{\mathrm{l}}}

\newcommand{\floor}[1]{\lfloor#1 \rfloor}

\newcommand{\dist}{\mathop{\mathrm{dist}}}
\renewcommand{\Im}{\mathop{\mathrm{Im}}}
\renewcommand{\Re}{\mathop{\mathrm{Re}}}

\newcommand{\clos}{\mathop{\mathrm{clos}}}
\newcommand{\inte}{\mathop{\mathrm{int}}}
\newcommand{\beq}{\begin{equation}}
\newcommand{\eeq}{\end{equation}}
\newcommand{\bal}{\begin{align}}
\newcommand{\eal}{\end{align}}
\newcommand{\nn}{\nonumber}

\newcommand{\ga}{\gamma}

\newcommand{\si}{\sigma}
\newcommand{\pa}{\partial}

\newcommand{\la}{\lambda}
\newcommand{\ov}{\overline}

\DeclareMathOperator{\wronsk}{\textup{\textsf{W}}}

\numberwithin{equation}{section}

\begin{document}

\title[On the KdV Equation with Steplike Finite-Gap Initial Data]{On the Cauchy Problem for the
Korteweg--de Vries Equation with Steplike Finite-Gap Initial Data II. Perturbations with Finite Moments}

\author[I. Egorova]{Iryna Egorova}
\address{B. Verkin Institute for Low Temperature Physics\\
47 Lenin Avenue\\61103 Kharkiv\\Ukraine}
\email{\mailto{iraegorova@gmail.com}}

\author[G. Teschl]{Gerald Teschl}
\address{Faculty of Mathematics\\ University of Vienna\\
Nordbergstrasse 15\\ 1090 Wien\\ Austria\\ and\\ International
Erwin Schr\"odinger
Institute for Mathematical Physics\\ Boltzmanngasse 9\\ 1090 Wien\\ Austria}
\email{\mailto{Gerald.Teschl@univie.ac.at}}
\urladdr{\url{http://www.mat.univie.ac.at/~gerald/}}

\thanks{Research supported by the Austrian Science Fund (FWF)
 under Grant No.\ Y330.}
\thanks{J. d'Analyse Math. {\bf 115:1}, 71--101 (2011)}


\begin{abstract}
We solve the Cauchy problem for the Korteweg--de Vries equation with
steplike quasi-periodic, finite-gap initial conditions under the assumption
that the perturbations have a given number of derivatives with finite moments.
\end{abstract}

\maketitle

\section{Introduction}

The purpose of this paper paper is to investigate the Cauchy problem for the
Korteweg--de Vries (KdV) equation
\beq\label{KdV}
q_t(x,t) = -q_{xxx}(x,t) + 6 q(x,t) q_x(x,t), \qquad q(x,0)= q(x),
\eeq
(where subscripts denote partial derivatives as usual) for the case of real-valued
steplike initial conditions $q(x)$. More precisely, we will assume
that $q(x)$ is asymptotically close to (in general) different
quasi-periodic, finite-gap potentials $p_\pm(x)$ in the sense that
\beq\label{2.111}
\pm \int_0^{\pm \infty} \left|
\frac{d^n}{dx^n}\big( q(x) - p_\pm(x)\big) \right| (1+|x|^{m_0})dx <\infty,
\quad 0\leq n\leq n_0,
\eeq
for some positive integers $m_0, n_0$. Here by quasi-periodic, finite-gap potentials
we mean algebro-geometric, quasi-periodic, finite-gap potentials which arise naturally as
the stationary solutions of the KdV hierarchy as discussed in \cite{GH} (further details will
be given in Section~\ref{secist}). If \eqref{2.111} holds for all $m_0, n_0$ we  call $q$ a
Schwartz-type perturbation.

Ever since the seminal work of Gardner et al.\ \cite{GGKM} in 1967 the
inverse scattering transform has become one of the main tools used for solving this Cauchy problem.
Numerous articles have been devoted to this subject since the publication of the GGKM paper.
In particular, the case that the initial condition is asymptotically close to $p_\pm(x)=0$
is well understood. We  refer to the monographs by Eckhaus and Van Harten
\cite{EVH}, Marchenko \cite{M},  Novikov, Manakov, Pitaevskii, and Zakharov \cite{NMPZ},
or Faddeev and Takhtajan \cite{FT}.

 There are two natural cases which have been
considered in the past when extending this classical situation. The first case is that of equal quasi-periodic, finite-gap
potentials $p_-(x)=p_+(x)$ and the second is the case of steplike constant asymptotics
$p_\pm(x) = c_\pm$ (with $c_- \neq c_+$). The aim of our present paper is to
combine both cases and to solve some open problems in these special
cases (to be discussed in detail below) along the way.

The underlying scattering theory in the case of asymptotically periodic solutions
was first investigated by Firsova \cite{F1}--\cite{F3}.
The first ones to consider the Cauchy problem with a periodic background
seem to be Kuznetsov and Mikha\u\i lov \cite{kumi}, who
informally treated the Korteweg--de Vries equation with the Weierstra{\ss}
elliptic function as background solution. It turns out, due to the poles of the
Baker--Akhiezer functions, which reflect the fact that the underlying hyperelliptic
Riemann surface is no longer simply connected, that the periodic case is much
more complicated. The only known results concerning the
existence of the solution seem to be by Ermakova \cite{Er}, \cite{Er1} and
Firsova \cite{F4} (where the evolution of the scattering data for periodic background
was given). However, both works are incomplete from the point of view of
a rigorous application of the inverse scattering method. Surprisingly, much
more is know about the asymptotical behavior
(assuming existence) of such solutions; see for example \cite{Ba},
\cite{Bik}--\cite{Bik2}, \cite{iu}, \cite{Kh}--\cite{KhS}, \cite{N}. A complete
and rigorous treatment of the inverse scattering transform for the KdV equation in the case
of initial conditions which are Schwartz-type perturbations of finite-gap solutions
was given only recently by Grunert and the present authors \cite{EGT}.

Let us now turn to the case of steplike constant potentials, $p_\pm(x)=c_\pm$. The foundations
for scattering theory are completely understood and were given in Buslaev and Fomin \cite{BF},
Davies and Simon \cite{DS}, Cohen and Kappeler \cite{CK0}, Gesztesy \cite{G}, and
Aktosun \cite{A}.

The corresponding Cauchy problem for the KdV equation was first investigated by Khruslov
 \cite{Kh}, who derives the time evolution of the scattering data and analyzes the
long-time asymptotics. Later, Cohen \cite{C} solved the case that $q(x)$ is the Heaviside step function.
Kappeler \cite{Kap}, based on some advances in scattering theory of Cohen and Kappeler \cite{CK1}, showed how to handle general initial conditions with only a fixed number
of moments finite.
The time evolution of the scattering data for the entire KdV hierarchy was computed
recently by Khasanov and Urazboev \cite{KU}.
However, while Kappeler's result is impressive from a technical point of view, it still does not
give a satisfactory answer, since it only determines the decay properties of the solution
near one side, whereas only very mild
information is given  concerning the decay properties at the other side. In particular, even if one starts with a Schwartz-type initial condition,
the results in \cite{Kap} do not guarantee that the solution stays within this
class. The reason for this is that \cite{Kap} (as well as \cite{C}) does not use the full
inverse scattering machinery but only a half-sided approach. For further results, where
the initial condition is supported on a half-line, see Rybkin \cite{ry} and
the references therein. The case of power-like
asymptotic behavior (including some unbounded initial conditions) was investigated by
Bondareva and Shubin \cite{Bo}, \cite{BS};see also \cite{KPST} for the case of the mKdV equation.
Finally, we mention that in the discrete steplike finite-gap case (Toda lattice), the same
problem was completely solved in \cite{EMT1}. For  analysis of the  corresponding
long-time asymptotic behavior, see \cite{EBM}, \cite{dkkz}, \cite{km2}, \cite{kt},
\cite{kt2}, \cite{krt2} and \cite{vdo}.

To state our main result we  denote the spectra of
the one-dimensional finite-gap Schr\"{o}dinger operators $L_\pm = - \pa_x^2 + p_\pm$
associated with the potentials $p_\pm(x)$ by
\beq
\sigma_\pm = [E_0^\pm, E_1^\pm]\cup\dots\cup[E_{2j-2}^\pm,
E_{2j-1}^\pm]\cup\dots\cup[E_{2r_\pm}^\pm,\infty).
\eeq
The various possible locations of the two spectra are illustrated in the  following example.
%-------------------%
\vskip 0.2cm\noindent
{\bf Example.}
Let $L_+$ be the two-band operator with spectrum
$\sigma_+=[E_1, E_2]\cup[E_4, +\infty)$ and $L_-$ the
three band operator with spectrum $\sigma_-=[E_1, E_2]\cup[E_3,
E_4]\cup [E_5,+\infty)$, where $E_1<E_2<\cdots<E_5$.
%-------------------%
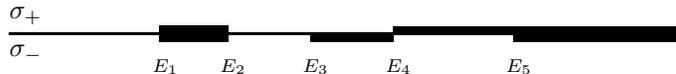
\begin{figure}[ht]
\begin{picture}(11,1.2)
\put(1,0.2){$\sigma_-$} \put(1,0.7){$\sigma_+$}

\put(1,0.5){\line(1,0){8}} \put(3,0.4){\rule{9mm}{2mm}}
\put(2.9,0){$\scriptstyle E_1$} \put(3.8,0){$\scriptstyle E_2$}
\put(5,0.4){\rule{11mm}{1mm}} \put(4.9,0){$\scriptstyle E_3$}
\put(6,0){$\scriptstyle E_4$}
\put(6.1,0.5){\rule{39mm}{1mm}}
\put(7.7,0.4){\rule{23mm}{1mm}} \put(7.6,0){$\scriptstyle E_5$}
\end{picture}
\caption{Typical locations of $\sigma_-$ and
$\sigma_+$.}\label{figsi}
\end{figure}

To shed some additional light on this we recall that the Marchenko kernel $F_\pm(x,y,t)$
(cf.\ \eqref{Fhat}) consists of three summands $F_\pm(x,y,t)= F_{\pm,D}(x,y,t) + F_{\pm,H}(x,y,t)
+ F_{\pm,R}(x,y,t)$, the first summand being a sum over all eigenvalues, the second one an integral over
$\si_\mp\backslash\si_\pm$, and the last one an integral over $\si_\pm$. The crucial part is
to show decay properties of $F_\pm(x,y,t)$ (and its partial derivatives). The first term
$F_{\pm,D}(x,y,t)$ is as nice as one can wish for and can thus be ignored. To get the
necessary decay for the remaining two terms one needs to use integration by parts.
In the classical case (and more general at
points $E_1$ and $E_2$ in our example), the corresponding boundary terms arising during
integration by parts will vanish and the required decay follows. However in a steplike
situation (at points $E_4$ and also $E_5$ for the $-$ case), this is no longer true.
Moreover, at a point like $E_5$ the integrand of $F_{+,R}$ has a non-differentiable singularity
which prevents an immediate integration by parts. By just working with the other kernel
one can evade this obstacle at the price of loosing the information about decay of the solution
at this side. Clearly these problems evaporate if points like $E_4$ and $E_5$ in our
example are absent. This was the case analyzed in \cite{EGT}. (In \cite{EGT} are also found
some necessary technical ingredients, which we  use freely here). Note that while this
restriction (which says that the respective spectral bands either coincide
or are disjoint) excludes the steplike constant case, it includes the case of short-range
perturbations of arbitrary quasi-periodic, finite-gap solutions.

It is the aim of the present paper to overcome these problems. To this end, rather than looking
at the terms $F_{\pm,H}(x,y,t)$ and $F_{\pm,R}(x,y,t)$ individually, we will in fact show that the
boundary terms mutually cancel. While this sounds like a pretty straightforward strategy,
this cancellation is by no means obvious and is nothing short of a small miracle.
Interestingly enough, points like $E_4$ and $E_5$ (the first one being absent in the
steplike constant case) turn out to require somewhat different miracles,  the first point
being more involved.

Our result  settles the aforementioned open problem of
steplike constant Schwartz-type perturbations as a special case. Moreover, based on
the recent advances in inverse scattering theory with steplike quasi-periodic, finite-gap backgrounds in
\cite{BET} (cf.\ also \cite{GNP}) and our preparations in \cite{EGT}, we are able to handle
not only steplike constant but also arbitrary steplike quasi-periodic, finite-gap backgrounds. For more on the
history of this problem and additional literature, see \cite{EGT}.

Next, let us state our main result.
Denote by $C^n(\R)$ the set of functions $x\in\R \mapsto q(x)\in \R$ which have $n$ continuous
derivatives with respect to $x$ \footnote{here $C^0(\mathbb R)= C(\mathbb R)$}  and by $C^n_k(\R^2)$ the set of functions $(x,t)\in\R^2 \mapsto q(x,t)\in \R$
which have $n$ continuous derivatives with respect to $x$ and $k$ continuous derivatives with respect to $t$.

\begin{theorem}\label{theor1}
Let $p_\pm(x,t)$ be two real-valued, quasi-periodic, finite-gap solutions of the KdV equation
corresponding to arbitrary quasi-periodic, finite-gap initial data $p_\pm(x)=p_\pm(x,0)$.
Let $m_0\geq 8$ and $n_0\geq m_0+5$ be fixed natural numbers.
Suppose that $q(x)\in C^{n_0}(\R)$ is a real-valued function such that
\eqref{2.111} holds. Then there exists a unique classical solution $q(x,t)\in C^{n_0 - m_0-2}_1(\R^2)$
of the initial-value problem for the KdV equation \eqref{KdV} satisfying
\beq \label{1.71}
\pm \int_0^{\pm \infty} \left| \frac{\pa^n}{\pa x^n} \big( q(x,t) -
p_\pm(x,t)\big) \right| (1+|x|^{\floor{\frac{m_0}{2}} - 2})dx
<\infty, \quad n \leq n_0 - m_0 -2,
\eeq
and
\beq \label{1.71t}
\pm \int_0^{\pm \infty} \left| \frac{\pa}{\pa t} \big( q(x,t) -
p_\pm(x,t)\big) \right| (1+|x|^{\floor{\frac{m_0}{2}} - 2})dx
<\infty,
\eeq
for all $t\in\R$.
\end{theorem}

In particular, this theorem shows that the KdV equation
has a solution in the class of steplike Schwartz-type perturbations of finite-gap potentials:

\begin{corollary}
Let $p_\pm(x,t)$ be two real-valued, quasi-periodic, finite-gap solutions of the KdV equation
corresponding to arbitrary quasi-periodic, finite-gap initial data $p_\pm(x)=p_\pm(x,0)$.
In addition, suppose, that $q(x)$ is a steplike Schwartz-type perturbation of $p_\pm(x)$.
Then the solution $q(x,t)$ of the initial-value problem for the KdV equation \eqref{KdV}
is a steplike Schwartz-type perturbation of $p_\pm(x,t)$ for all $t\in\R$.
\end{corollary}

The above results can also be used to solve  analogous Cauchy problems
for the modified KdV equation \cite{ET3}.
Furthermore, it might also be of independent interest that for uniqueness the following weaker requirement is
sufficient.

\begin{theorem}
Let $p_\pm(x,t)$ be two real-valued, quasi-periodic, finite-gap solutions of the KdV equation
corresponding to arbitrary quasi-periodic, finite-gap initial data $p_\pm(x)=p_\pm(x,0)$.
Suppose $q(x,t)$ is a solution of the KdV Cauchy problem satisfying
\beq\label{conduniq}
\pm \int_0^{\pm \infty} \left( |q(x,t) - p_\pm(x,t)| + \left| \frac{\pa}{\pa t} \big( q(x,t) -
p_\pm(x,t)\big) \right| \right) (1+x^2)dx <\infty,
\eeq
then $q(x,t)$ is unique in this class of solutions.
\end{theorem}

\begin{proof}
The assumption are sufficient to prove the time evolution of the scattering data
\cite[Lemma~5.3]{EGT}. Moreover, by \cite[Corollary 4.4]{BET} the scattering
data uniquely determine $q(x,t)$ and the claim follows.
\end{proof}

\section{The inverse scattering transform for the KdV equation with steplike
finite-gap initial data}
\label{secist}

In \cite{EGT}, we  established the inverse scattering transform for the KdV
equation in the case of Schwartz-type perturbations. In this section, we
review the necessary steps and identify the changes required  for the
present, more general, situation. These changes are implemented in
the next section. For further information and for the history of
finite-gap solutions see, for example, \cite{GH}, \cite{GRT},
\cite{M}, or \cite{NMPZ}. For further information on the underlying scattering
theory and its history see \cite{BET}.

To set the stage let
\beq
L_\pm(t) =-\frac{d^2}{d x^2} +p_\pm(x,t)
\eeq
be two one-dimensional Schr\"{o}dinger operators, corresponding to
two real-valued, quasi-periodic, finite-gap solutions $p_\pm(x,t)$ of the KdV equation that are associated
with the spectra
\beq\label{1.61}
\sigma_\pm = [E_0^\pm, E_1^\pm]\cup\dots\cup[E_{2j-2}^\pm,
E_{2j-1}^\pm]\cup\dots\cup[E_{2r_\pm}^\pm,\infty)
\eeq
and the Dirichlet divisors
\beq\label{divisor}
\left\{\big(\mu_1^\pm(t),\si_1^\pm(t)\big), \dots, \big(\mu_{r_\pm}^\pm(t),
\si_{r_\pm}^\pm(t)\big)\right\},
\eeq
respectively. Here we assume without loss of
generality that all gaps are open, that is, $E_{2j-1}^\pm<E_{2j}^\pm$ for
$j=1,2,...,r_\pm$. We will abbreviate $\mu_j^\pm(0)=\mu_j^\pm$,
$\si_j^\pm(0)=\si_j^\pm$.

Let us cut the complex plane along the spectrum $\sigma_\pm$ and
denote the upper and lower sides of the cuts by $\sipmu$ and
$\sipml$. Denote the corresponding points of the cuts by
$\lau$ and $\lal$, respectively. In particular, this means
\beq
f(\lau) := \lim_{\varepsilon\downarrow0} f(\la+\I\varepsilon),
\qquad f(\lal) := \lim_{\varepsilon\downarrow0}
f(\la-\I\varepsilon), \qquad \la\in\sigma_\pm.
\eeq
Set
\beq\label{1.0}
Y_\pm(\la)=-\prod_{j=0}^{2r_\pm} (\la-E_j^\pm),
\eeq
and
introduce the functions
\beq\label{1.88}
g_\pm(\la,t)=
-\frac{\prod_{j=1}^{r_\pm}(\la - \mu_j^\pm(t))}{2 Y_\pm^{1/2}(\la)},
\eeq
where the branch of the square root is chosen such that
\beq\label{1.8}
\frac{1}{\I} g_\pm(\lau,t) = \Im(g_\pm(\lau,t))  >0
\quad \mbox{for}\quad \la\in\sigma_\pm,\quad t\in\R_+.
\eeq
Denote by
\beq\label{psin}
\psi_\pm(\la,x,t)=c_\pm(\la,x,t)+ m_\pm(\la,t)s_\pm(\la,x,t)
\eeq
the  Weyl solutions of the equations
\beq\label{eqqp}
L_\pm(t)y=\la y,
\eeq
normalized according to $\psi_\pm(\la,0,t)=1$ and such that $\psi_\pm(\la,\cdot,t)\in
L^2(\R_\pm)$ for $\la\in\C \setminus\si_\pm$. Here, $m_\pm(t)$ are
the Weyl functions and $c_\pm(\la,x,t)$ and $s_\pm(\la,x,t)$ are
solutions of \eqref{eqqp}, that satisfy the initial conditions
\beq
c_\pm(\la,0,t)=s_\pm^\prime(\la,0,t)=1, \quad s_\pm(\la,0,t)=
c_\pm^\prime(\la,0,t)=0.
\eeq
The functions $\psi_\pm$ admit the well-known representation
\beq\label{1.23}
\psi_\pm(\la,x,t)=u_\pm(\la,x,t)\E^{\pm\I\theta_\pm(\la)x},
\quad\la\in\C\setminus\si_\pm,
\eeq
where  $\theta_\pm(\la)$ are the
quasimomenta and the functions $u_\pm(\la,x,t)$ are quasiperiodic
with respect to $x$ with the same basic frequencies as the
potentials $p_\pm(x,t)$. The quasimomenta are holomorphic for
$\la\in\C\setminus\si_\pm$ and normalized according to
\beq\label{1.24}
\frac{d\theta_\pm}{d\la}>0 \quad
\mbox{for}\quad\la\in\sipmu,\qquad \theta_\pm(E_0^\pm)=0.
\eeq
This normalization implies
\beq\label{1.25}
\frac{d\theta_\pm}{d\la}=\frac{\I\prod_{j=1}^{r_\pm}(\la -
\zeta_j^\pm)}{ Y_\pm^{1/2}(\la)},\qquad \zeta_j^\pm\in(E_{2j-1}^\pm,
E_{2j}^\pm),
\eeq
and therefore, the quasimomenta are real-valued on
$\sipmul$. Note, that in the case where $p_\pm(x,t)\equiv0$ we have
$\theta_\pm(\la)=\sqrt{\la}$, $u_\pm(\la,x,t)\equiv 1$ and
$m_\pm(\la,t)=\pm\I\sqrt\la$. In the general finite-gap cases the
two Weyl $m$-functions associated with $L_\pm$ are given by
(\cite[eq.~(1.165)]{GH})
\beq\label{1.29}
m_\pm(\la,t)=\frac{H_\pm(\la,t)\pm
Y_\pm^{1/2}(\la)}{\prod_{j=1}^{r_\pm}(\la - \mu_j^\pm(t))}, \quad
\breve m_\pm(\la,t)=\frac{H_\pm(\la,t)\mp
Y_\pm^{1/2}(\la)}{\prod_{j=1}^{r_\pm}(\la - \mu_j^\pm(t))}.
\eeq
 Here $H_\pm(\la,t)$ are polynomials in $\la$ of
$\deg(H_\pm)\leq r_\pm-1$ with real-valued coefficients which are smooth
with respect to $t$. Moreover,
\beq\label{ischez}
H_\pm(\mu_j^\pm(t),t)=0\quad\mbox{for}\quad
\mu_j^\pm(t)\in\pa\si_\pm.
\eeq
 Associated with the second Weyl
$m$-function $\breve m_\pm(\la,t)$ is the second Weyl solution
\begin{align}\nn
\breve\psi_\pm(\la,x,t) &= c_\pm(\la,x,t)+\breve m_\pm(\la,t) s_\pm(\la,x,t)\\\label{1.30}
&=  \breve u_\pm(\la,x,t)\E^{\mp\I\theta_\pm(\la)x},
\quad\la\in\C\setminus\si_\pm,
\end{align}
that satisfies $\breve\psi_\pm(\la,\cdot,t)\in L^2(\R_\mp)$ for $\la\in\C
\setminus\si_\pm$. The Wronski determinant, $\wronsk(f,g)(x)=f(x)g'(x)-f'(x)g(x)$, of the functions $\psi_\pm$ and
$\breve\psi_\pm$ is given by
\beq\label{wrpsi}
\wronsk(\psi_\pm(\la,.,t),\breve\psi_\pm(\la,.,t)) = \pm g_\pm(\la,t)^{-1}.
\eeq
Introduce the Lax operators corresponding to the finite-gap solutions $p_\pm(x,t)$,
\begin{align}\label{Lop}
L_\pm(t) &= -\pa_x^2 + p_\pm(x,t),\\
P_\pm(t) &=  -4\pa_x^3 + 6p_\pm(x,t)\pa_x +3 \pa_x p_\pm(x,t).
\end{align}
Then the following result is valid (\cite{BBEIM}, \cite{GH})
\begin{lemma}\label{lemweyl1}
The functions
\beq\label{1.37}
\hat\psi_\pm(\la,x,t) = \E^{\alpha_\pm(\la,t)} \psi_\pm(\la,x,t),
\eeq
where
\beq\label{1.38}
\alpha_\pm(\la,t) := \int_0^t \left(2(p_\pm(0,s) + 2\la)
m_\pm(\la,s) - \frac{\pa p_\pm(0,s)}{\pa x}\right)ds,
\eeq
satisfy the system of equations
\begin{align}\label{LP1}
L_\pm(t)\hat\psi_\pm &= \la\hat\psi_\pm,\\ \label{LP2}
\frac{\pa\hat\psi_\pm}{\pa t} &= P_\pm(t)\hat\psi_\pm.
\end{align}
\end{lemma}

Set
\begin{align}\label{Mset}
M_\pm(t) &=\{ \mu^\pm_j(t) \mid \mu^\pm_j(t) \in
(E_{2j-1}^\pm,E_{2j}^\pm) \text{ and } m_\pm(\la,t) \text{ has a
simple pole}\},\\ \nn \hat M_\pm(t) &=\{ \mu^\pm_j(t) \mid
\mu^\pm_j(t) \in \{E_{2j-1}^\pm, E_{2j}^\pm\} \},
\end{align}
and introduce the functions
\begin{align} \nn
\delta_\pm(\la,t) &:= \prod_{\mu^\pm_j(t) \in M_\pm(t)}(\la-\mu^\pm_j(t)),\\
\label{S2.6}
\hat \delta_\pm(\la,t) &:= \prod_{\mu^\pm_j(t) \in M_\pm(t)}
(\la-\mu_j^\pm(t)) \prod_{\mu^\pm_j(t) \in \hat M_\pm(t)} \sqrt{\la -
 \mu^\pm_j(t)},
\end{align}
where $\prod =1$ if the index set is empty.
These functions  allow us to remove the singularities of the Weyl solutions
$\psi_\pm(\la,x,t)$ whenever necessary.

Next, we collect now some facts from scattering theory for Schr\"odinger operators
with smooth steplike finite-gap potentials (cf.\ \cite{BET}, \cite{EGT}).
To shorten notations throughout this discussion, we omit the dependence on $t$.

Let $n_1\geq 0$ and $m_1\geq 2$ be given natural numbers and let
 $q(x)\in C^{n_1}(\mathbb R)$ be a real-valued  function such that
\beq\label{S.2}
\pm \int_0^{\pm \infty}\left|\frac{d^n}{d x^n}\big(q(x) -
p_\pm(x)\big)\right| (1+|x|^{m_1})dx <\infty,
 \quad  \forall\, 0\leq n\leq n_1.
\eeq
Consider the {\em perturbed} operator
\beq\label{S.12}
L :=- \frac{d^2}{dx^2} +q(x)
\eeq
with a potential $q(x)$ that satisfies \eqref{S.2}. The spectrum of $L$
consists of a purely absolutely continuous part
$\sigma:=\sigma_+\cup\sigma_-$, plus a finite number of discrete eigenvalues
$\sigma_d=\{\la_1,\dots,\la_p\}$ situated in the gaps
$\sigma_d\subset\R\setminus\sigma$.  The set
$\sigma^{(2)}:=\sigma_+\cap\sigma_-$ is the spectrum of multiplicity
two for the operator $L$, and the set $\sigma_+^{(1)}\cup\sigma_-^{(1)}$ with
$\sigma_\pm^{(1)}= \clos(\sigma_\pm\setminus\sigma_\mp)$ is the
spectrum of multiplicity one.

%-------------------%
The Jost solutions of the spectral equation
\beq\label{S.4}
\left(-\frac{d^2}{dx^2}+q(x)\right) \phi(x)= \la \phi(x),\quad \la\in \C,
\eeq
are defined by the requirement that they asymptotically look like the
Weyl solutions of the background operators as $x\to\pm\infty$.

\begin{lemma}\label{lemJost5}
Assume $q(x)$ satisfies \eqref{S.2}. Then there exist solutions
$\phi_\pm(\la, x)$, $\la \in \C$, of \eqref{S.4} satisfying
\beq
\phi_\pm(\la,x) = \psi_\pm(\la,x) (1 + o(1)), \qquad x\to\pm\infty.
\eeq
The Jost solutions $\phi_\pm(\la, .)$ are meromorphic with respect to $\la\in\C\backslash\si_\pm$
and have the same poles as $\psi_\pm(\la, .)$. The functions $\hat\delta(\la)\phi_\pm(\la, \cdot)$ are continuous up to the
boundary $\siu_\pm\cup\sil_\pm$. Moreover, $\hat\delta(\la)\phi_\pm(\la,\cdot)$ are $m_1$ times differentiable
with respect to $\la \in \inte(\siu_\pm\cup\sil_\pm)$ and $m_1-1$ times continuously differentiable with
respect to the local variable $\sqrt{\la - E}$ near $E\in\pa\si_\pm$.
\end{lemma}

\begin{proof}
Set
\beq\label{J}
J_\pm(\la,x,y)=\frac{\psi_\pm(\la,y)\breve\psi_\pm(\la,x) -
\psi_\pm(\la,x)\breve\psi_\pm(\la,y)}{\wronsk(\psi_\pm(\la),\breve\psi_\pm(\la))}.
\eeq
Then the Jost solutions of \eqref{S.4} formally satisfy the  integral equation
\beq
\phi_\pm(\la,x) = \psi_\pm(\la,x) -
\int_x^{\pm\infty} J_\pm(\la,x,y) (q(y) - p_\pm(y)) \phi_\pm(\la,y)dy.
\eeq
To remove the singularities of $\psi_\pm(\la,x)$ near $\la\in M_\pm \cup \hat M_\pm$,
one can multiply the whole equation by $\hat\delta_\pm(\la)$.

Similarly, the $x$ derivatives satisfy
\[
\frac{\pa}{\pa x} \phi_\pm(\la,x) = \frac{\pa}{\pa x} \psi_\pm(\la,x) -
\int_x^{\pm\infty} \left( \frac{\pa}{\pa x} J_\pm(\la,x,y) \right) (q(y) - p_\pm(y)) \phi_\pm(\la,y)dy.
\]
Hence, existence of the Jost solutions together with their derivatives follows from
existence of solutions of these integral equations. Existence is proved by the method of
successive iterations in the usual manner. Observe that since at points $\la\in\pa\si_\pm$
the second solution grows linearly, the above kernel can only be estimated by $C|x-y|$
near such points.
\end {proof}

We  also need to know the asymptotic behavior of the Jost solutions as $\la\to\infty$.
To determine it, we recall the well-know expansion (cf.\ the proof of Lemma~1.19 in \cite{GH})
\beq\label{exppsi}
\psi_\pm(\la,x)=\exp\left(\pm\I \sqrt{\la} x
+\int_0^x\left(\sum_{j=1}^n \frac{\kappa^\pm_j(y)}
{(\pm 2\I \sqrt{\la})^j} +\frac{\tilde\kappa_{n,\pm}(\sqrt{\la},y)}
{(\pm 2\I \sqrt{\la})^n}\right)dy\right),
\eeq
up to any order $n$, where
\beq\label{expcpsi}
\kappa^\pm_1(x)=p_\pm(x),\quad
\kappa^\pm_{j+1}(x)=- \frac{\pa}{\pa x}\kappa^\pm_j(x)-\sum_{i=1}^{j-1}
\kappa^\pm_{j-i}(x)\kappa^\pm_i(x),
\eeq
and the error term satisfies
\beq
\frac{\pa^l }{\pa k^l}\tilde\kappa_{n,\pm}(k,x)=o(1), \quad  l=0,1,\dots
\eeq
for fixed $x$ as $k\to \infty$.

\begin{lemma}\label{lemasymphi}
Assume $q(x)$ satisfies \eqref{S.2}. Then the Jost solutions have the asymptotic expansions
\begin{equation}
\phi_\pm(\la,x) =   \psi_\pm(\la,x) \left( 1 + \frac{\phi_{\pm,1}(x)}{\la^{1/2}} + \dots + \frac{\phi_{\pm,n_1+1}(x)}{\la^{(n_1+1)/2}} +
o\big(\la^{-(n_1+1)/2} \big)\right)
\end{equation}
which can be differentiated $m_1$ times with respect to $\la^{1/2}$. An analogous
expansion holds for $\frac{\pa}{\pa x} \phi_\pm(\la,x)$.
\end{lemma}

\begin{proof}
To obtain the asymptotic expansion, consider $\tilde{\phi}_\pm(\la,x)= \frac{\phi_\pm(\la,x)}{\psi_\pm(\la,x)}$,
which satisfy
\beq
\tilde{\phi}_\pm(\la,x) = 1 -
\int_x^{\pm\infty} \tilde{J}_\pm(\la,x,y) (q(y) - p_\pm(y)) \tilde{\phi}_\pm(\la,y)dy,
\eeq
where $\tilde{J}_\pm(\la,x,y) = J_\pm(\la,x,y) \frac{\psi_\pm(\la,y)}{\psi_\pm(\la,x)}$. Next recall
\eqref{1.23}, \eqref{1.30}, and \eqref{wrpsi} which imply
\[
\tilde{J}_\pm(\la,x,y) = \pm g_\pm(\la) \left( u_\pm(\la,y)^2 \frac{\breve u_\pm(\la,x)}{u_\pm(\la,x)}
\E^{\pm 2\I\theta_\pm(\la)(x-y)} - \breve u_\pm(\la,y) u_\pm(\la,y)\right),
\]
where $u_\pm(\la,x)$, $\breve u_\pm(\la,x)$ are quasi-periodic with respect to $x$ and have
convergent expansions around $\infty$ with respect to $\theta_\pm(\la)^{-1}$. Now use the fact that
\[
\int_0^\infty \E^{2 \I \theta(\la) y} f(\la,y) dy = \sum_{j=1}^n \frac{f_j}{\theta(\la)^j} + o(\theta(\la)^{-n})
\]
provided $f(\la,x)$ is $n$ times differentiable with respect to $x$, the first $n-1$ derivatives
have an asymptotic expansion with respect to $\theta(\la)^{-1}$ of order $n$ and the $n$'th
derivative satisfies $\lim_{\la\to\infty} \frac{\pa^n}{\pa x^n} f(\la,x) = g(x)$ in $L^1(0,\infty)$.
This follows from $n$ partial integrations and the Riemann-Lebesgue Lemma (cf.\
also \cite[Theorem~3.2]{MT}).

As
in the previous lemma, the claims for the derivatives follow by considering the corresponding integral equations.
\end{proof}

\begin{corollary}\label{corasymphi}
Assume $q(x)$ satisfies \eqref{S.2}. Then the Weyl $m$-functions $m_{q,\pm}(\la,x)= \frac{\phi_\pm'(\la,x)}{\phi_\pm(\la,x)}$
have the asymptotic expansion
\beq
m_{q,\pm}(\la,x) =  \pm \I \sqrt{\la} + \sum_{j=1}^{n_1} \frac{\kappa_j(x)}{(\pm 2\I\sqrt{\la})^j} + o(\la^{-n_1/2}),
\eeq
which can be differentiated $m_1$ times with respect to $\la^{1/2}$. The coefficients $\kappa_j(x)$ are
given by \eqref{expcpsi} with $q(x)$ in place of $p_\pm(x)$.
\end{corollary}

\begin{proof}
Existence of the expansion follows from the previous lemma, and the expansion coefficients follow
by comparing coefficients in the Riccati equation
\[
\frac{\pa}{\pa x} m_{q,\pm}(\la,x) + m_{q,\pm}(\la,x)^2 + \la - q(x) =0.
\]
\end{proof}

The Jost solutions can be represented, with the help of the transformation operators, as
\beq\label{S2.2}
\phi_\pm(\la,x) =\psi_\pm(\la,x)\pm\int_{x}^{\pm\infty} K_\pm(x,y)
\psi_\pm(\la,y) dy,
\eeq
where $K_\pm(x,y)$ are real-valued functions that satisfy
\beq\label{A.5}
 K_\pm(x,x)=\pm\frac{1}{2}\int_x^{\pm\infty} (q(y)-p_\pm(y))dy.
\eeq
Moreover, as a consequence of \cite[(A.15)]{BET}, we have the estimate
\beq\label{S2.3}
\left|\frac{\pa^{n+l}}{\pa x^n\pa
y^l} K_\pm(x,y)\right|\leq C_\pm(x)\left(Q_\pm(x+y)
+\sum_{j=0}^{n+l-1} \left|\frac{\pa^j}{\pa x^j}\big(q(\frac{x+y}{2})
-p_\pm(\frac{x+y}{2})\big)\right|\right),
\eeq
for $\pm y>\pm x$, where $C_\pm(x)=C_{n,l,\pm}(x)$ are continuous positive
functions decaying as $x\to\pm\infty$, and
\beq\label{S2.333}
Q_\pm(x):=
\pm\int_{\frac{x}{2}}^{\pm\infty} \big|q(y) - p_\pm(y)\big| dy.
\eeq
Formula \eqref{S2.2} shows that the Jost solutions inherit all
singularities of the background  Weyl functions $m_\pm(\la)$ and
Weyl solutions $\psi_\pm(\la)$. In particular, as a direct
consequence of formulas \eqref{psin}, \eqref{1.25}, \eqref{1.29}, \eqref{ischez},
\eqref{S2.2}, and Lemma~\ref{lemJost5}
we have the following result.

\begin{lemma}\label{lemMhat}
Let $E\in\pa\si_\pm$ and let $\varepsilon>0$ be such that
$[E-\varepsilon, E+\varepsilon]\cap\pa\si_\pm=\{E\}$ and
$\varepsilon <\dist(\mu_j^\pm,E)$ if $\mu_j^\pm\neq E$.

\noindent
{\bf(i)} Let $\mu_j^\pm= E$.
Introduce the functions
\beq\label{phige}
\phi_{\pm,E}(\la,x):=
\I\left(\theta_\pm(\la)-\theta_\pm(E)\right)\phi_\pm(\la,x),\
g_{\pm,E}(\la):=
\left|\theta_\pm(\la)-\theta_\pm(E)\right|^{-2}g_\pm(\la)
\eeq
for $\la [E-\varepsilon, E+\varepsilon]$. The functions $\phi_{\pm,E}(\la,x)$ admit
the representation
\beq\label{struct3}
\phi_{\pm,E}(\la,x)=c_{\pm,E}(\la,x) +\I
\left(\theta_\pm(\la)-\theta_\pm(E)\right)s_{\pm,E}(\la,x),
\eeq
where $c_{\pm,E}(\cdot,x), s_{\pm,E}(\cdot,x)\in C^{m_0-1}([E-\varepsilon, E+\varepsilon])$
and $c_{\pm,E}(\cdot,x), s_{\pm,E}(\cdot,x)\in \mathbb R$.
Analogous representations hold for $\frac{\pa}{\pa x} \phi_{\pm,E}(\la,x)$.
Moreover,
\beq\label{realval}
\phi_{\pm,E}(\la,x)\in\mathbb{R},\qquad \la\in
[E-\varepsilon, E+\varepsilon]\setminus\si_\pm,
\eeq
and 
\beq\label{xiwronsk}
g_{\pm,E}(\la)^{-1}=\pm\wronsk(\phi_{\pm,E},
\overline{\phi_{\pm,E}}), \qquad \la\in (E-\varepsilon,
E+\varepsilon)\cap\si_\pm.
\eeq
{\bf (ii)} Let $\mu_j^\pm\neq E$.
The functions $\phi_\pm(\la,x)$  admit the same representation
\eqref{struct3} on the set $[E-\varepsilon/2, E+\varepsilon/2]$.
\end{lemma}

Next, recalling \eqref{S2.6}, set
\beq\label{S2.12}
\tilde\phi_\pm(\la,x)=\delta_\pm(\la) \phi_\pm(\la,x)
\eeq
so that the functions $\tilde\phi_\pm(\la,x)$ have no poles in the interior
of the gaps of the spectrum $\si$. For each eigenvalue $\la_k$, we
introduce the corresponding norming constants
\beq \label{S2.14}
\left(\gamma_k^\pm\right)^{-2}=\int_{\R} \tilde\phi_\pm^2(\la_k,x)
dx.
\eeq
Furthermore, recall the scattering relations
\beq\label{S2.16}
T_\mp(\la) \phi_\pm(\la,x)
=\overline{\phi_\mp(\la,x)} + R_\mp(\la)\phi_\mp(\la,x),
\quad\la\in\simpul,
\eeq
where the transmission and reflection
coefficients are defined as usual, by
\beq\label{2.17} T_\pm(\la):=
\frac{\wronsk(\overline{\phi_\pm(\la)},
\phi_\pm(\la))}{\wronsk(\phi_\mp(\la), \phi_\pm(\la))},\qquad
R_\pm(\la):= -
\frac{\wronsk(\phi_\mp(\la),\overline{\phi_\pm(\la)})}
{\wronsk(\phi_\mp(\la), \phi_\pm(\la))}, \quad\la\in \sipmul.
\eeq

%-------------------%
\begin{lemma}\label{lem2.3}
Suppose that $q(x)\in C^{n_1}(\mathbb R)$ satisfies \eqref{S.2}. Then the scattering data
\begin{align}\nn
{\mathcal S} = \Big\{ & R_+(\la),\;T_+(\la),\;
\la\in\sigma_+^{\mathrm{u,l}}; \; R_-(\la),\;T_-(\la),\;
\la\in\sigma_-^{\mathrm{u,l}};\\\label{S4.6} &
\la_1,\dots,\la_p\in\R\setminus \sigma,\;
\gamma_1^\pm,\dots,\gamma_p^\pm\in\R_+\Big\}
\end{align}
has the following properties:
\begin{enumerate}[\bf I.]
\item
\begin{enumerate}[\bf(a)]
\item
$T_\pm(\lau) =\overline{T_\pm(\lal)}$ for $\la\in\sigma_\pm$.\\
$R_\pm(\lau) =\overline{R_\pm(\lal)}$ for $\la\in\sigma_\pm$.
\item
$\dfrac{T_\pm(\la)}{\overline{T_\pm(\la)}}= R_\pm(\la)$
for $\la\in\sigma_\pm^{(1)}$.
\item
$1 - |R_\pm(\la)|^2 =
\dfrac{g_\pm(\la)}{g_\mp(\la)} |T_\pm(\la)|^2$ for
$\la\in\sigma^{(2)}$.
\item
$\overline{R_\pm(\la)}T_\pm(\la) +
R_\mp(\la)\overline{T_\pm(\la)}=0$ for
$\la\in\sigma^{(2)}$.
\item
$T_\pm(\la) = 1 + O\Big(\frac{1}{\sqrt\la}\Big)$ for
$\la\to\infty$.
\item
$R_\pm(\la) = o\Big(\frac{1}{\left(\sqrt{\la}\right)^{n_1+1}}\Big)$ for
$\la\to\infty$.
\end{enumerate}
\item
The functions $T_\pm(\la)$ can be extended as meromorphic
functions to the domain $\C \setminus \sigma$ and satisfy
\beq\label{S2.18}
\frac{1}{T_+(\la) g_+(\la)} = \frac{1}{T_-(\la) g_-(\la)}=:-W(\la),
\eeq
where $W(\la)$ possesses the following properties:
\begin{enumerate}[\bf(a)]
\item
The function $\tilde W(\la)=\delta_+(\la)\delta_-(\la) W(\la)$ is
holomorphic in the domain $\C\setminus\sigma$, with simple
zeros at the points $\la_k$, where
\beq\label{S2.11}
\frac{d\tilde W}{d \la}(\la_k) = (\gamma_k^+\gamma_k^-)^{-1}.
\eeq
In addition, it satisfies
\beq\label{S2.9}
\overline{\tilde W(\lau)}=\tilde W(\lal), \quad
\la\in\sigma\quad \text{and}\quad \tilde W(\la)\in\R
\quad \text{for} \quad \la\in\R\setminus \sigma.
\eeq
\item
The function $\hat W(\la) = \hat\delta_+(\la) \hat\delta_-(\la)
W(\la)$ is continuous on the set $\C\setminus\sigma$ up to the
boundary $\siu\cup\sil$. Moreover, this function is $m_1-1$ times
differentiable with respect to $\la$ on the set
$\left(\siu\cup\sil\right)\setminus (\pa\si_-\cup\pa\si_+)$ and
$m_1-1$ times continuously differentiable with respect to the local
variable $\sqrt{\la - E}$ for $E\in\pa\si_-\cup\pa\si_+$. It can have
zeros on the set $\pa\sigma_-\cap\pa\si_+$ and does not vanish
at the other points of the set $\sigma$.  If $\hat W(E)=0$ for
$E\in\pa\sigma_-\cap\pa\si_+$, then $\hat W(\la) = \sqrt{\la -E}
(C(E)+o(1))$, $C(E)\ne 0$.
\end{enumerate}
\item
\begin{enumerate}[\bf(a)]
\item
The reflection coefficients $R_\pm(\la)$ are continuous functions on $\siu_\pm\cup\sil_\pm$.
They are also $m_1$ times differentiable with respect to $\la$ on the sets
$\sipmul\setminus\{\pa\si_+\cup\pa\si_-\}$ and
$m_1-j$ times differentiable with respect to the coordinate $\sqrt{\la-E}$ with
$E\in\{\pa\si_+\cup\pa\si_-\}\cap\sipmul$, where $j=1$ if $\hat W(E)\neq 0$ and
$j=2$ if $\hat W(E)\neq 0$.
The asymptotics {\bf I.~(f)} hold for all derivatives as well.
\item
If $E\in\pa\si_\pm$ and $\hat W(E)\neq 0$, then
\beq\label{P.1}
R_\pm(E)=
\begin{cases}
-1 &\text{for } E\notin\hat M_\pm,\\
1 &\text{for } E\in\hat M_\pm.
\end{cases}
\eeq
\end{enumerate}
\end{enumerate}
\end{lemma}

\begin{proof}
Except for {\bf I.~(f)} and the corresponding statement for the derivatives in {\bf III.~(a)} everything follows
as in \cite[Lemma~4.1]{EGT}. To prove the missing items, we need to prove that $\wronsk(\phi_\pm,\ov{\phi_\mp})=o(\la^{-n_1/2})$
and all its necessary derivatives with respect to $\sqrt\la$. But this follows from
\[
\wronsk(\phi_\pm,\ov{\phi_\mp}) = \phi_-(\la,x) \ov{\phi_+(\la,x)} ( m_{q,-}(\la,x) - \ov{m_{q,+}(\la,x)})
\]
since $\phi_-(\la,x) \ov{\phi_+(\la,x)}=O(1)$ by Lemma~\ref{lemasymphi} and $m_{q,-}(\la,x) - \ov{m_{q,+}(\la,x)}=o(\la^{-n_1/2})$
by Corollary~\ref{corasymphi}.
\end{proof}

Next, recall the associated Gelfand--Levitan--Marchenko (GLM) equations
\beq\label{ME}
K_\pm(x,y) + F_\pm(x,y) \pm \int_x^{\pm\infty} K_\pm(x,\xi)
F_\pm(\xi,y)d\xi =0, \quad \pm y>\pm x,
\eeq
where\footnote{Here we have used the notation
$\oint_{\sigma_\pm}f(\la)d\la := \int_{\sipmu} f(\la)d\la -
\int_{\sipml} f(\la)d\la$.}
\begin{align}\label{4.2}
F_\pm(x,y) &= \frac{1}{2\pi\I}\oint_{\sigma_\pm}
R_\pm(\la) \psi_\pm(\la,x) \psi_\pm(\la,y)
g_\pm(\la)d\la + \\ \nn &\quad + \frac{1}{2\pi
\I}\int_{\sigma_\mp^{(1),\mathrm{u}}} |T_\mp(\la)|^2
\psi_\pm(\la,x) \psi_\pm(\la,y)g_\mp(\la)d\la\\ \nn
&\quad + \sum_{k=1}^p (\gamma_k^\pm)^2 \tilde\psi_\pm(\la_k,x)
\tilde\psi_\pm(\la_k,y).
\end{align}
As in \cite[Lemma~4.2]{EGT}, we have the following result:

\begin{lemma}
Under the same assumptions as in Lemma~\ref{lem2.3}, the functions $F_\pm(x,y)$ satisfy
\begin{enumerate}[\bf I.]
\addtocounter{enumi}{3}
\item
$F_\pm(x,y)\in C^{(n_1+1)}(\R^2)$. There exist real-valued continuous functions $\tilde q_{\pm}(x)$, with
$x^{m_1}\tilde q_{\pm}\in L^1(\mathbb R_\pm)$, and monotone positive continuous functions
$C_\pm(x)$, $Q_\pm(x)$, which decay as $x\to\pm\infty$, with $x^{m_1-1}Q_\pm(x)\in L^1(\mathbb R_\pm)$,
such that  for $\pm x>\pm a$, $\pm y >\pm a$ and $0\leq n+l\leq n_1+1$  the inequalities hold
\beq \label{4.4}
\left|\frac{\pa^{n+l}}{\pa x^n\pa y^l}
F_\pm (x,y)\right|\leq C_\pm(a) \left( Q_\pm(x+y) +\tilde q_\pm(x+y)(1-\delta_{n+l,0})\right).
\eeq
Here $\delta_{n,m}$ is the Kronecker delta and $a\in\mathbb R$ is an arbitrary fixed number.
Moreover,
\beq \label{4.5}
\pm\int_0^{\pm\infty} \left|\frac{d^{n}}{dx^{n}} F_\pm
(x,x)\right|(1+ |x|^{m_1}) dx <\infty,\qquad 1\leq n\leq n_1+1.
\eeq
\end{enumerate}
\end{lemma}

\begin{proof}
The GLM equation \eqref{ME} and \eqref{4.2} are derived in \cite{BET}.
Estimate \eqref{4.4} follows directly from \eqref{ME} and \eqref{S2.3}.
Equation \eqref{4.5} is then immediate from \eqref{4.4}.
\end{proof}

As demonstrated in \cite{BET} and \cite{EGT}, 
properties \textbf{I--IV} are necessary and sufficient for a
set $\mathcal{S}$ to be the set of scattering data for operator $L$
with a potential $q(x)$ satisfy \eqref{S.2}.

Now the procedure of solving of the inverse scattering problem is as follows:

Let $L_\pm$ be two one-dimensional finite-gap Schr\"odinger
operators associated with the potentials $p_\pm(x)$. Let
$\mathcal{S}$ be given data as in \eqref{S4.6} satisfying
\textbf{I--IV}. Define corresponding kernels $F_\pm(x,y)$ via
\eqref{4.2}. As shown in \cite{BET}, under  condition
\textbf{IV} the GLM equations \eqref{ME} have unique smooth
real-valued solutions $K_\pm(x,y)$, that satisfy estimates of type
\eqref{4.4}. In particular,
\beq\label{5.101}
\pm\int_0^{\pm\infty} (1+|x|^{m_1})\left|\frac{d^n}{dx^n}
K_\pm(x,x)\right|
dx<\infty, \qquad 1\leq n\leq n_1+1.
\eeq
Now introduce the functions
\beq\label{5.1}
q_\pm(x) = p_\pm(x) \mp 2\frac{d}{dx}K_\pm(x,x)
\eeq
and note that  \eqref{5.101} reads
\beq\label{5.2}
\pm \int_0^{\pm \infty}\left| \frac{d^n}{dx^n} \big( q_\pm(x) -
p_\pm(x)\big) \right| (1+|x|^{m_1}) d x <\infty ,\quad 0\leq n\leq n_1.
\eeq
We obtain the following result.

\begin{theorem}[\cite{BET}]\label{theor4}
Let the set of data ${\mathcal S}$, defined as in \eqref{S4.6},
satisfy  properties \textbf{I--IV}. Then the functions $q_\pm(x)$ defined by \eqref{5.1}
satisfy \eqref{5.2} and coincide, $q_-(x)\equiv q_+(x)=:q(x)$. Moreover, the set ${\mathcal S}$
is the set of scattering data for the Schr\"odinger operator \eqref{S.12} with
potential $q(x)$ satisfying \eqref{S.2}.
\end{theorem}

Our next step is to describe a formal scheme for using the inverse scattering method to solve the
initial value problem for the KdV equation with initial conditions $q(x)$
satisfying \eqref{2.111} with with some quasi-periodic, finite-gap potentials $p_\pm(x)$ and
fixed $m_0\geq 8$ and $n_0\geq m_0+5$. Consider the
corresponding scattering data $\mathcal{S}=\mathcal{S}(0)$ which
obey conditions \textbf{I--IV} with $n_1=n_0$ and
$m_1=m_0$. Let $p_\pm(x,t)$ be the finite-gap solution of the KdV
equation with initial conditions $p_\pm(x)$ and let $m_\pm(\la,t)$,
$\breve m_\pm(\la,t)$ $\psi_\pm(\la,x,t)$,  $\alpha_\pm(\la,t)$ be
defined by \eqref{1.29}, \eqref{psin} and \eqref{1.38} as above. Also set
\beq\label{1.38new}
\breve\alpha_\pm(\la,t)=\int_0^t
\left(2(p_\pm(0,s) + 2\la) \breve m_\pm(\la,s) - \frac{\pa
p_\pm(0,s)}{\pa x}\right)ds.
\eeq
Introduce the set  $\mathcal S(t)$ by
\begin{align}\nn
{\mathcal S}(t) = \Big\{ & R_+(\la,t),\;T_+(\la,t),\;
\la\in\sigma_+^{\mathrm{u,l}}; \;
R_-(\la,t),\;T_-(\la,t),\;
\la\in\sigma_-^{\mathrm{u,l}};\\\label{S4.6t} &
\la_1,\dots,\la_p\in\R\setminus \sigma,\;
\gamma_1^\pm(t),\dots,\gamma_p^\pm(t)\in\R_+\Big\},
\end{align}
 where $\la_k(t)$, $R_\pm(\la,t)$, $T_\pm(\la,t)$ and
$\gamma_k^\pm(t)$ are defined (\cite[Lemma~5.3]{EGT}) by:
\begin{align}
\label{refl} R_\pm(\la,t) &= R_\pm(\la,0)\E^{\alpha_\pm(\la,t)
-\breve{\alpha}_\pm(\la,t)}, \quad \la\in\si_\pm, \\ \label{trans}
T_\mp(\la,t) &= T_\mp(\la,0)\E^{\alpha_\pm(\la,t)
-\breve{\alpha}_\mp(\la,t)},\quad\la\in\mathbb C,\\ \label{norm}
\left(\gamma_k^\pm(t)\right)^2 &= \left(\gamma_k^\pm(0)\right)^2
\frac{\delta_\pm^2(\la_k,0)}{\delta_\pm^2(\la_k,t)}
\E^{2\alpha_\pm(\la_k,t)},
\end{align}
where $\alpha_\pm(\la,t)$, $\breve{\alpha}_\pm(\la,t)$,
$\delta_\pm(\la,t)$ are defined in \eqref{1.38}, \eqref{1.38new},
\eqref{S2.6} respectively.

In \cite{EGT}, it is proved, that these data satisfy \textbf{I--III} with $g_\pm(\la,t)$,
defined by \eqref{1.88} and $\delta_\pm(\la)$, $\hat\delta_\pm(\la)$ defined
by \eqref{S2.6}.

Introduce 
\begin{align}\label{6.2}
F_\pm(x,y,t) =& \frac{1}{2\pi\I}\oint_{\si_\pm}
R_\pm(\la,t) \psi_\pm(\la,x,t) \psi_\pm(\la,y,t) g_\pm(\la,t)d\la +
\\ \nn & {} + \frac{1}{2\pi\I}\int_{\si_\mp^{(1),\mathrm{u}}}
 |T_\mp(\la,t)|^2
\psi_\pm(\la,x,t) \psi_\pm(\la,y,t)g_\mp(\la,t)d\la \\ \nn & {} +
\sum_{k=1}^p (\ga_k^\pm(t))^2 \tilde\psi_\pm(\la_k,x,t)
\tilde\psi_\pm(\la_k,y,t).
\end{align}
Suppose that we are able to prove that $F_\pm$ satisfy
\beq\label{4.311}
\left|\frac{\pa^{n+l}}{\pa x^n\pa y^l}
F_\pm(x,y,t)\right|+\left|\frac{\pa^2} {\pa x\pa t}
F_\pm(x,y,t)\right|\leq \frac{C}{|x+y|^{m_1+2}}\quad
n+l\leq n_1+1,
\eeq
as $x,y\to\pm\infty$ for some $m_1\geq 2$, $n_1\geq 3$, and
$C=C(n_1,m_1,t)$. Then \eqref{4.311} implies that condition {\bf IV} 
holds with $Q_\pm(x)= (1+|x|^{m_1+2})^{-1}$, $\tilde q_\pm(x)=0$, and $C_\pm(a)$,
that exists due to continuity of functions $F_\pm(x,y,t)$ together with their derivatives.  Thus Theorem~\ref{theor4}
ensures the unique solvability of the time dependent GLM equations
\beq\label{ME1}
K_\pm(x,y,t) + F_\pm(x,y,t) \pm \int_x^{\pm\infty}
K_\pm(x,\xi,t) F_\pm(\xi,y,t) d\xi =0, \quad \pm y>\pm x,
\eeq
and yields the function
\beq\label{5.111}
q(x,t) = p_\pm(x,t) \mp 2 \frac{d}{dx}K_\pm(x,x,t).
\eeq
By construction $q$ satisfies (cf.\ \eqref{5.2})
\beq\label{S.2t}
\pm \int_0^{\pm \infty} \left| \frac{\pa^n}{\pa x^n} \big( q(x,t) -
p_\pm(x,t)\big) \right| (1+|x|^{m_1})dx <\infty,\quad 0 \leq n\leq
n_1,
\eeq
and, as in \cite{EGT}, one concludes that \eqref{4.311} also implies differentiability with
respect to $t$ such that
\beq\label{Deriv t}
\pm \int_0^{\pm \infty} \left| \frac{\pa}{\pa t} \big( q(x,t) - p_\pm(x,t)\big) \right|
(1+|x|^{m_1})dx <\infty.
\eeq
Moreover, by following the arguments in Section~6 of \cite{EGT} verbatim
(see, in particular, Lemma~6.3 and Corollary~2.3) one establishes that $q(x,t)$
solves the associated initial-value problem of the KdV equation. Thus, to prove
Theorems~\ref{theor1}, it is sufficient to prove the inequality \eqref{4.311}
with $m_1= \floor{\frac{m_0}{2}}-2$, $n_1=n_0 -m_0 -2$.

\section{Proof of the main result}

To obtain \eqref{4.311} we follow the approach, developed in
\cite{EGT}. First of all, recall that the functions $F_\pm(x,y,t)$
are given by
\begin{align}\label{Fhat}
F_\pm(x,y,t) =& \frac{1}{2\pi\I}\oint_{\si_\pm}
R_\pm(\la,0) \hat\psi_\pm(\la,x,t) \hat\psi_\pm(\la,y,t)
g_\pm(\la,0)d\la +
\\ \nn & {} + \frac{1}{2\pi\I}\int_{\si_\mp^{(1),\mathrm{u}}}
 |T_\mp(\la,0)|^2
\hat\psi_\pm(\la,x,t)\hat \psi_\pm(\la,y,t)g_\mp(\la,0)d\la \\ \nn &
{} + \sum_{k=1}^p (\ga_k^\pm(0))^2 \breve\psi_\pm(\la_k,x,t)
\breve\psi_\pm(\la_k,y,t),
\end{align}
where $\hat\psi_\pm(\la,x,t)$ are defined by \eqref{1.37} and we have set
\beq\label{hattil}
\breve\psi_\pm(\la,x,t):= \delta_\pm(\la,0)\hat\psi_\pm(\la,x,t).
\eeq
Furthermore, recall that the functions $\hat\psi_\pm(\la,x,t)$ inherit their
singularities from $\psi_\pm(\la,x,0)$; that is, they have simple poles on 
$M_\pm(0)$ and square-root singularities  $\hat M_\pm(0)$.
Consequently, the functions \eqref{hattil} are bounded and smooth in small
vicinities of the points $\la_k$. Moreover, all integrands in \eqref{Fhat}
have only integrable singularities (cf.\ \cite[Sect.~5]{BET}) and thus
all three summands in \eqref{Fhat} are well defined. Our aim is to study the decay of $F_\pm(x,y,t)$ as $x$, $y$
tend to $\pm\infty$, respectively.

First of all, we observe, that the third summand in \eqref{Fhat}
(corresponding to the discrete spectrum)  together with all its derivatives decays exponentially as
$x+y\to\pm\infty$. Therefore, it
satisfies \eqref{4.311} for all natural $m_1$ and $n_1$.  In the
second summand,  $\hat\psi_\pm(\la,x,t)\hat
\psi_\pm(\la,y,t)$ together with all derivatives decays exponentially
with respect to $(x+y)\to\pm\infty $ for $\la\notin\si_\pm$. Hence we have to
estimate this summand only in small vicinities of the points
$\si_\pm\cap\si_\mp^{(1)}$.

Our strategy is as follows. In both integrals
of \eqref{Fhat} we make a change of variables
from $\la$ to the quasimomentum variables $\theta_\pm$  and use \eqref{1.23} to represent the integrands as
$\E^{\pm\theta_\pm (x+y)}\rho_\pm(\la(\theta_\pm),x,y,t)$, where  $\rho_\pm$ together with their derivatives are smooth and uniformly bounded with respect
to $x,y\in\R$. Moreover, since these functions are differentiable with respect to $\theta_\pm$
(and also bounded with respect to $x$ and $y$),  we will integrate by parts both integrals in \eqref{Fhat}
as many times as possible and then prove that the boundary terms either cancel or vanish.

To investigate the validity of integration by parts for the
first summand in \eqref{Fhat} we use \eqref{1.23}--\eqref{1.25} to
represent the first summand as
\begin{align}\nn
F_{\pm,R}(x,y,t) &:= 2\Re \int_{\si_\pm^{\mathrm{u}}}
R_\pm(\la,t) \psi_\pm(\la,x,t)
\psi_\pm(\la,y,t) \frac{g_\pm(\la,t)}{2\pi\I}d\la\\ \label{Fc}
&=\Re\int_0^\infty \E^{\pm\I(x+y)\theta_\pm}
\rho_\pm(\theta_\pm,x,y,t) d\theta_\pm,
\end{align}
where
\beq\label{defrho}
\rho_\pm(\theta_\pm,x,y,t) := \frac{1}{2\pi} R_\pm(\la,0) u_\pm(\la,x,t) u_\pm(\la,y,t)
\E^{2\alpha_\pm(\la,t)}
\prod_{j=1}^{r_\pm}\frac{\la -\mu_j^\pm}{\la-\zeta_j^\pm},
\eeq
with $\la=\la(\theta_\pm)$.
Since the integrand in \eqref{Fc} is not
continuous at $\theta_\pm(E_{2k+1}^\pm)=\theta_\pm(E_{2k+2}^\pm)$, we regard this integral as
\beq\label{sumFc}
F_{\pm,R}(x,y,t)=
\Re \sum_{k=0}^{r_\pm +1} \int_{\theta_\pm(E_{2k}^\pm)}^
{\theta_\pm(E_{2k+1}^\pm)} \E^{\pm\I (x+y) \theta}
\rho_\pm(\theta,x,y,t) d\theta,
\eeq
where we have set
\[
E_{2r_\pm+1}^\pm=E_{2r_\pm+2}^\pm=\tilde E>\max\{E_{2r_+}^+,
E_{2r_-}^-\},
\]
and $E_{2r_\pm+3}^\pm=+\infty$ for notational convenience.

The boundary terms arising from integration by parts (except for the last one, corresponding to
$+\infty$) become
\beq\label{outint}
\Re \lim_{\la\to E}\frac{\E^{\pm\I\theta_\pm(E)(x+y)}\frac{\pa^s \rho_\pm(\theta_\pm,
x,y,t)}{\pa\theta_\pm^s}}{\left(\I (x+y)\right)^{s+1}},\quad
E\in\pa\si_\pm\cup\tilde E,\: s=0,1,\dots,m.
\eeq
The number $m$ of possible integrations by parts is directly related to the smoothness
of $R_\pm(\la,0)$ and thus the to the values of $m_0$ and $n_0$.
To estimate the boundary terms in \eqref{sumFc} we distinguish three cases:

\begin{enumerate}
\item[1)]
$E\in\pa\si_\pm\cap\pa\si$ (points $E_1$, $E_2$ in our example and
also point $E_3$ for $F_{-,R}(x,y,t)$);
\item[2)]
$E\in\pa\si_\pm\cap\inte(\si_\mp)$ (the point $E_5$ for $F_{-,R}(x,y,t)$);
\item[3)]
$E\in\pa\si_-^{(1)}\cap\pa\si_+^{(1)}$ (the point $E_4$).
\end{enumerate}

In the first case, the boundary terms \eqref{outint} will vanish. In the second and the third
cases, however, these terms do not vanish, but we will prove, that they
cancel with a corresponding terms from the second summand in \eqref{Fhat}.
Finally, the two boundary terms stemming from our artificial boundary
point $\tilde E$ will cancel and hence do not need to be taken into account.

The following result which takes care of 1), is an immediate consequence of the proof of
\cite[Lemma~6.2]{EGT}.

\begin{lemma} \label{lemestim3}
Let $E\in\pa\si_\pm\cap\pa\si$. Then the
following limits exists and assume either real or purely imaginary values:
\beq\label{reflest} \lim_{\la\to
E,\,\la\in\si_\pm}\E^{\pm\I\theta_\pm(E)(x+y)}
\frac{\pa^s}{\pa\theta_\pm^s} \rho_\pm(\theta_\pm,x,y,t)\in \I^s\R,
\eeq
for $s=0,\dots, m_0-1$ if $\hat W(E)\ne 0$ and $s=0,\dots, m_0-2$
if $\hat W(E)= 0$.
\end{lemma}

This lemma shows, that the boundary terms \eqref{outint} vanish at the
points corresponding to case 1). Before turning to the cases 2) and 3)
let us first start by discussing smoothness of the integrand
$\rho_\pm(\theta,x,y,t)$ in \eqref{sumFc}.

Since except for $R_\pm(\la,0)$, all other parts of $\rho_\pm(\theta_\pm,x,y,t)$
are smooth with respect to $\la\in\inte(\si_\pm)$, it suffices to look
at $R_\pm(\la,0)$. By Lemma~\ref{lem2.3}, {\bf III.~(a)} the latter function
has $m_0$ derivatives with respect to $\la$ (and consequently also
with respect to $\theta_\pm$) as long as we stay in the interior of $\si_\pm$
and away from boundary points of $\si_\mp$. Hence no such points
 pose any problems; the only problematic points are
those in $\pa\si_\mp \cap \inte(\si_\pm)$ (the point $E_5$ in our
example for $F_{+,R}(x,y,t)$). Hence we will address this issue first.

Let $E\in\pa\si_\mp \cap \inte(\si_\pm)$ be such a point.
As already pointed out, only $R_\pm(\la,0)$ matters and by
Lemma~\ref{lem2.3}, {\bf III.~(a)} we can write it locally as a
smooth function of $\sqrt{\la-E}$. Thus we obtain
\beq\label{sing}
\frac{\pa^s\rho_\pm(\theta_\pm,x,y,t)}{\pa\theta_\pm^s}=
O\left(\frac{1}{\sqrt{(\la-E)^{2s-1}}}\right).
\eeq
Since this singularity is non-integrable for $s\ge 2$,
integration by parts is not an option near such points. Hence we
ill split off the leading behavior near such a point. The leading
term near each such point can be computed explicitly and the
remainder can be handled by integration by parts.

Since the last interval $(\tilde E,\infty)$ does not contain such points
we can restrict our attention to finite intervals. Moreover, for notational
convenience we will restrict ourselves to the case of $F_{+,R}$.

Abbreviate $\theta=\theta_+$ and denote by
\[
E_i\in \pa\si_-\cap\left(E_{2j}^+,E_{2j+1}^+\right),\quad i=1,\dots,N,
\]
our {\em bad} points. Let $\varepsilon>0$ and introduce the cutoff functions
\beq\label{func10}
B_i(\theta):=B(\frac{\theta-\theta(E_i)}{\varepsilon}),\quad i=1,\dots,N,
\eeq
where
\beq\label{funk1}
B(\xi)=\begin{cases}
\E^{-\xi^2}\left(1 - \xi^{2m_0}\right)^{m_0}, & \mbox{for }
|\xi|\leq 1,\\
0, & \mbox{for } |\xi|\geq 1.
\end{cases}
\eeq
We  choose $\varepsilon>0$ so small  that the supports of the
functions $B_i(\theta)$  neither intersect nor  contain small
vicinities of the points $\theta(E_{2j}^+)$ and $\theta(E_{2j+1}^+)$.
Moreover, we have
\beq\label{funk2}
\aligned
\frac{d^s B_i}{d\theta^s}(\theta(E_i)\pm\varepsilon)=0,,\: s=0,\dots,m_0-1, \\
\frac{d^s B_i}{d\theta^s}(\theta(E_i))=0,\: s=1,\dots,2m_0+1.
\endaligned
\eeq
Now we can rewrite the $j$-th summand in \eqref{sumFc} (except for the last one) as
\begin{align*}
& \int_{\theta(E_{2j}^+)}^ {\theta(E_{2j+1}^+)}
\E^{\I (x+y) \theta} \rho_+(\theta,x,y,t) d\theta=\\
&\qquad = \int_{\theta(E_{2j}^+)}^{\theta(E_{2j+1}^+)}
\E^{\I (x+y) \theta} \left(1-\sum_{i=1}^N
B_i(\theta)\right)\rho_+(\theta,x,y,t) d\theta + \\
&\qquad\quad +\sum_{i=1}^N \int_{-\infty}^ {\infty}
\E^{\I (x+y) \theta}
B_i(\theta)\rho_+(\theta,x,y,t) d\theta.
\end{align*}
Because of \eqref{funk2} the first term can be integrated by parts $m_0$ times
and thus is covered by Lemma~\ref{lemestim3}. For the second
term, we switch to the local variable $z=\sqrt{\theta - \theta(E_i)}$
and use a Taylor expansion for the integrand,
\[
\rho_+(\theta,x,y,t)=\rho_0^{(i)}(x,y,t) + \rho_1^{(i)}(x,y,t) z
+\dots+\rho_{m_0-1}^{(i)}(x,y,t) z^{m_0-2} + \beta_i(\theta),
\]
where $\beta_i(\theta) = O\left(z^{m_0-1}\right)$
has $\floor{\frac{m_0}{2}}$ integrable derivatives with respect to
$\theta$ in a small vicinity of the point $\theta(E_i)$. By construction
\[
\frac{\pa^s (B_i\beta_i)}{\pa\theta^s}
(\theta(E_i)\pm\varepsilon)=0,\quad
s=0,\dots,\floor{\frac{m_0}{2}},
\]
and thus 
\beq\label{funk7}
\int_{-\infty}^ {\infty} \E^{\I (x+y) \theta} B_i(\theta)\beta_i(\theta) d\theta=
O\left( (x+y)^{-\floor{\frac{m_0}{2}}}\right).
\eeq
To compute the remaining terms, observe that
\begin{align*}
& \int_{-\infty}^ {\infty} \E^{\I (x+y) \theta}
B_i(\theta)\left(\sqrt{\theta - \theta(E_i)}\right)^\nu d\theta=\\
& \quad = (\varepsilon)^{\nu/2+1} \E^{\I(x+y)\theta(E_i) }
\int_{-1}^1 \E^{-\zeta^2+ \I \varepsilon(x+y) \zeta }
\left(1 - \zeta^{2m_0}\right)^{m_0}\zeta^{\nu/2}d\zeta,
\end{align*}
and note that we can extend the integral over the interval $(1,1)$ to the interval$(-\infty,\infty)$, since
\[
\int^{\pm\infty}_{\pm 1} \E^{-\zeta^2 + \I \varepsilon(x+y) \zeta}
\left(1 - \zeta^{2m_0}\right)^{m_0}\zeta^{\nu/2}d\zeta
= O\!\left( (x+y)^{-m_0-1}\right).
\]
Now we  simply expand
\[
\left(1 -
\zeta^{2m_0}\right)^{m_0}= 1 -m_0 \zeta^{2 m_0} + \cdots +(-1)^{m_0}
\zeta^{2m_0^2}
\]
and evaluate the integral by invoking the integral representation \cite[9.241]{GR}\footnote{It
also follows from 3.462 3, but this formula contains a sign error.}
for the parabolic cylinder functions $\mathcal D_\kappa(z)$ (cf.\ \cite{GR}, \cite{WW}). This
 gives
\begin{align}\nn
&\int_{-\infty}^\infty \E^{\I \varepsilon (x+y) \zeta}
\E^{-\zeta^2} \zeta^{\kappa}d\zeta =\\ \label{kapp}
& \qquad =(-\I)^{\kappa} 2^{-\kappa/2} \sqrt\pi\exp\left(-\frac{\varepsilon^2
(x+y)^2}{8}\right)
\mathcal{D}_{\kappa}\left(\frac{\varepsilon(x+y)}{\sqrt 2}\right), \quad \Re(\kappa)>-1.
\end{align}
Since the parabolic cylinder functions have the following expansion \cite[9.246 1]{GR},
\[
\mathcal D_\kappa(z)\sim z^\kappa \E^{-\frac{z^2}{4}}\left(1 -
\frac{\kappa(\kappa -1)}{2 z^2} + \cdots\right), \quad |\arg(z)|<\frac{3\pi}{4},
\]
for large $z$,  the integral \eqref{kapp} decays exponentially as $(x+y)\to\infty$
for any $\kappa>0$. Combining these estimates
with Lemma~\ref{lemestim3} we obtain the following

\begin{lemma}\label{lemest4}
Let $E_{2j}^\pm, E_{2j+1}^\pm\in\pa\si_\pm\cap\pa\si$. Then
\beq\label{spec}
\frac{\partial^{n+l}}{\partial x^n\partial y^l}\Re
\int_{\theta_\pm(E_{2j}^\pm)}^{\theta_\pm(E_{2j+1}^\pm)} \E^{\pm\I
(x+y) \theta} \rho_\pm(\theta,x,y,t) d\theta = O\!\left(
(x+y)^{-\floor{\frac{m_0}{2}}}\right)
\eeq
as $x,y\to\pm\infty$ for all fixed $n,l=0,1,\dots$.
\end{lemma}

Note, that the condition $E_{2j}^\pm, E_{2j+1}^\pm\in\pa\si_\pm\cap\pa\si$ is only
used to take care of the boundary terms obtained from integration by parts and can hence
be replaced with any other condition which takes care of these terms.

Now we come to case 2) and study the behavior of the boundary terms
at the points $E\in\pa\si_\pm\cap\inte \si_\mp$. In this case formula
\eqref{reflest} remains valid only for $s=0$, so we need to take
the second summand in \eqref{Fhat} into account.

For notational convenience we  consider only the $+$ case
and assume, without loss of generality, that $E=E_{2j}^+$ . In this case,
$\si^{(2)}$ is located to the right of $E$ and $\si_-^{(1)}$ to the
left. Moreover, without loss of generality, we assume that
the other boundary terms are already covered by the previous
considerations so that we do not have to worry about them.

Choose $\varepsilon>0$ so small that
\[
[\la(\theta_+(E) + \I\varepsilon), E]\subset \left((\xi_j^+, E]\cap
\si_-^{(1)}\right),\quad (E,\la(\theta_+(E +\varepsilon))]\subset
\inte \si^{(2)}.
\]
Introduce in these two small intervals the two new (positive) variables
\beq\label{ash}
h:=\frac{\theta_+ - \theta_+(E)}{\I},\quad k:=\theta_+ - \theta_+(E).
\eeq
We  compare the boundary terms at the point $E$ for the two integrals:
\beq\label{Rint}
\Re \int_{\theta(E)}^ {\theta(E+\varepsilon)} \E^{\I (x+y) \theta_+}
\rho_+(\theta_+,x,y,t) d\theta_+ =
\Re\int_0^\varepsilon R(k)\Psi(\la(k),x,y,t) \E^{\I k (x+y)} dk
\eeq
and
\begin{align}\nn
&\int_{\la(\theta(E) + \I\varepsilon)}^E |T_-(\la,0)|^2
\hat\psi_+(\la,x,t)\hat \psi_+(\la,y,t)\frac{g_-(\la,0)}{2\pi\I}d\la =\\ \label{T}
& \qquad = \int_\varepsilon^0 P(h)\Psi(\la(h),x,y,t) \E^{-h(x+y)} dh,
\end{align}
with
\beq\label{rho1}
\Psi(\la,x,y,t)=\frac{\E^{\I\theta(E)(x+y)}}{2\pi}
\prod_{j=1}^{r_+}\frac{\la -\mu_j^+}{\la -\zeta_j^+} \E^{-\I
(x+y) \theta} \hat\psi_+(\la,x,t) \hat\psi_+(\la,y,t),
\eeq
and
\begin{align}\label{defR}
R(k) := &R_+(\la,0),\\\nn
P(h) :=& \frac{-\I}{2g_+(\la,0)g_-(\la,0) |W(\la,0)|^2}\\\label{defP}
=& \frac{-\I}{2g_+(\la,0)g_-(\la,0)\wronsk_0(\phi_-,\phi_+)
\wronsk_0(\overline{\phi_-},\phi_+)},
\end{align}
where $\wronsk_0(\cdot,\cdot)=\wronsk(\cdot,\cdot)|_{t=0}$.
Equation \eqref{defP} was obtained by using \eqref{S2.18} together with the fact that
$\overline{g_-(\la,0)} = -g_-(\la,0)$ if $\la\in\si_-$.

Integrating  \eqref{Rint} and \eqref{T} by parts with respect to $k$
and $h$, respectively, gives
\begin{align}\nn
& \int_{\varepsilon}^0 P(h)\Psi(\la(h)) \E^{-h(x-y)} dh =
-\sum_{j=0}^{m-1} \frac{1}{(x-y)^{j+1}}\frac{\pa^j (P \Psi)}{\pa h^j} (0) \\ \label{decompos}
& \qquad {} + \frac{1}{(x-y)^m}\int_\varepsilon^0 \frac{\pa^m (P\Psi)}{\pa h^m} \E^{-h(x-y)} dh
+ O(\E^{-\varepsilon(x-y)}) ,
\end{align}
\begin{align}\nn
& \Re\int_0^\varepsilon R(k)\Psi(\la(k)) \E^{\I k(x-y)} dk=
\Re \sum_{j=0}^{m-1} \frac{1}{(-\I (x-y))^{j+1}}\frac{\pa^j(R\Psi)}{\pa k^j}(0)\\ \label{decomposR}
& \qquad {} +  \Re \frac{1}{(-\I(x-y))^m}\int_0^\varepsilon\frac{\pa^{m}(R\Psi)}
{\pa k^m} \E^{\I k(x-y)} d k.
\end{align}
For the boundary terms to cancel each other we need
\beq\label{import}
\lim_{k\to 0}\Re \left(\I^{j+1} \frac{\pa^j(R\Psi)}{\pa k^j}(k)\right)=
\lim_{h\to 0}\frac{\pa^j(P\Psi)} {\pa h^j}(h), \quad j=0,\dots,m_0-1,
\eeq
where the left limit is taken from the side of the spectrum of multiplicity two and
the right limit is taken from the side of the spectrum of multiplicity one.
Since $\Psi$ is smooth to any degree with respect to $k$ and $h$ near $E$,
\beq
\lim_{k\to 0} \left(\I^j \frac{\pa^j \Psi}{\pa k^j}(k)\right)=
\lim_{h\to 0}\frac{\pa^j \Psi} {\pa h^j}(h), \quad j=0,\dots
\eeq
We observe that to prove \eqref{import}, it suffices to prove the following lemma.

\begin{lemma}\label{Lemma 3.3}
Let $h, k, P(h), R(k)$ be defined by \eqref{ash}, \eqref{defR}, and \eqref{defP}.
If $E\in\pa\si_\pm\cap\inte(\si_\mp)$, then
\beq\label{import1}
\lim_{k\to 0} \Re\left(\I^{j+1}\,\frac{d^j R(k)}{d k^j}\right)=
\lim_{h\to 0} \frac{d^j P(h)}{d h^j}, \qquad j=0,\dots,m_0-1.
\eeq
\end{lemma}

\begin{proof}
To prove this formula, recall that $\phi_-(\cdot,x)$,
$\ov{\phi_-(\cdot,x)}\in C^{m_0}(E-\varepsilon, E+\varepsilon)$ (and
similarly for the $x$ derivative)  since
$E\in\inte \si_-$. Therefore their derivatives with respect to
$\sqrt{\la - E}$ are smooth in a vicinity of $k=0$.  Without loss of
generality, we suppose\footnote{Otherwise replace $\phi_+(\la,x,0)$
by $\phi_{+,E}(\la,x,0)$ and $g_+(\la,0)$ by $g_{+,E}(\la,0)$  (cf.\ \eqref{phige})
in the subsequent considerations.},
that $E\neq \mu_j^+$, that is, the function $\phi_+(\la,x,0)$ as well as the
functions $g_+(\la,0)$ and $g_-(\la,0)$ (see \eqref{1.88}) are also
smooth with respect to $\sqrt{\la - E}$. For $\la>E$ introduce the function
\[
\tilde{P}(k):=  \frac{-\I}{2g_+(\la,0)g_-(\la,0)\wronsk_0(\phi_-,\phi_+)
\wronsk_0(\overline{\phi_-},\phi_+)}.
\]
Then
\beq\label{mainn}
\lim_{k\to +0} \I^s \frac{d^s \tilde{P}(k)}{d k^s}= \lim_{h\to +0}\frac{d^s P}{dh^s}.
\eeq
From $g_\pm(\la,0)^{-1}=\pm\wronsk_0(\phi_\pm,\overline{\phi_\pm})$ we
see that
\beq\label{def P1}
\tilde{P}(k)= \frac{\I\,\wronsk_0(\phi_-,\ov{\phi_-})\wronsk_0(\phi_+,\ov{\phi_+})}
{2\wronsk_0(\phi_-,\phi_+)
\wronsk_0(\ov{\phi_-},\phi_+)}.
\eeq
Substituting
\[
\phi_+(\la,x,0)=\frac{\wronsk_0(\phi_+,\ov{\phi_-})}
{\wronsk_0(\phi_-,\ov{\phi_-})} \phi_-(\la,x,0)
 -\frac{\wronsk_0(\phi_+,\phi_-)}{\wronsk_0(\phi_-,\ov{\phi_-})}
 \ov{\phi_-(\la,x,0)}
\]
into the numerator of \eqref{def P1} gives
\[
\tilde{P}(k)=\frac{\I}{2}\left(-\frac{\wronsk_0(\phi_-,\ov{\phi_+})}
{\wronsk_0(\phi_-,\phi_+)}+
\frac{\wronsk_0(\ov{\phi_-},\ov{\phi_+})}{\wronsk_0(\ov{\phi_-},\phi_+)}
\right).
\]
Introducing the abbreviations
\beq\label{wrons1}
W(k):=\wronsk_0(\phi_-,\phi_+), \quad
V(k):=\wronsk_0(\phi_-,\ov{\phi_+}).
\eeq
we thus have
\beq\label{P1}
R(k)= - \frac{V(k)}{W(k)}, \qquad \tilde{P}(k) =
\frac{\I}{2}\left(-\frac {V(k)}{W(k)} +\frac{\ov{W(k)}}{\ov{V(k)}} \right).
\eeq
Next, for all $x$  and  small $\varepsilon >0$, we have
\[
\phi_-(\la,x,0), \frac{\pa}{\pa x}\phi_-(\la,x,0) \in
C^{m_0}(E-\varepsilon, E+\varepsilon).
\]
Therefore, according to Lemma \ref{lemMhat} {\bf (ii)}, near $k=0$,
for positive $k$, we have the representation
\[
W(k) - V(k)= \I k f_1(k^2),\quad W(k) +V(k)=f_2(k^2),
\]
where $f_{1,2}(\cdot)\in C^{m_0-1}([0,\varepsilon_1))$. Differentiating
these relations gives
\beq\label{signs}
\lim_{k\to+0} \frac{\pa^s}{\pa k^s} V(k) =
(-1)^s\lim_{k\to+0} \frac{\pa^s}{\pa k^s} W(k),\qquad
s=0,\dots,m_0-1,
\eeq
and hence we see that $V(k)= W_{m_0-1}(-k)+o(k^{m_0-1})$, where
$W_{m_0-1}(k)$ is the Taylor polynomial of degree $m_0-1$ for $W(k)$.
Now recall
\[
\wronsk_0(\phi_-,\phi_+)(E)\neq 0,
\]
which implies that $R^{-1}(k)=R_{m_0-1}(-k) + o(k^{m_0-1})$, where $R_{m_0-1}(k)$
is the Taylor polynomial of degree $m_0-1$ for $R(k)$. Thus, we finally obtain
\beq
\tilde{P}(k)=\frac{\I}{2}\left( R_{m_0-1}(k) - \ov{R_{m_0-1}(-k)}\right) +o(k^{m_0-1}),
\eeq
from which \eqref{import1} follows.
\end{proof}

Lemma \ref{Lemma 3.3} settles case 2). Case 3) will follow from the next lemma.
\begin{lemma}
Let $h, k, P(h), R(k)$ be defined by \eqref{ash}, \eqref{defR}, and \eqref{defP}. Then,
if $E\in\pa\si_-^{(1)}\cap\pa\si_+^{(1)}$,
\beq
\lim_{k\to 0} \Re\left(\I^{j+1}\,\frac{d^j R(k)}{d k^j}\right)=
\lim_{h\to 0} \frac{d^j P(h)}{d h^j}, \qquad j=0,\dots,m_0-1.
\eeq
\end{lemma}

\begin{proof}
Note that now we cannot proceed as in case 2) since now we no longer have
spectrum of multiplicity two to the right of $E$. In particular, we cannot use
$g_-(\la,0)^{-1}=-\wronsk_0(\phi_-,\overline{\phi_-})$ for $\la>E$, we do
not have the scattering relations at our disposal, and $\phi_-(\la)\not\in
C^{m_0}(E-\varepsilon, E+\varepsilon)$. Hence we need a different strategy.

Let\footnote{Again, otherwise replace $\phi_\pm(\la,x,0)$ by $\phi_{\pm,E}(\la,x,0)$
and $g_\pm(\la)$ by $g_{\pm,E}(\la)$ (cf.\ Lemma~\ref{lemMhat}) in the subsequent considerations.}
$E\notin \hat M_-(0)\cup\hat M_+(0)$. Consider $\phi_\pm(\la,x,0)$ and note
that for sufficiently small $\varepsilon$, we can write (see Lemma~\ref{lemMhat})
\beq\label{defin5}
\phi_-(\la,x,0) = \begin{cases}
f^-_1(h^2,x) + \I h f^-_2(h^2,x) + o(h^{m_0-1}),
& E-\varepsilon<\la \le E,\\
f^-_1(-k^2,x) + k f^-_2(-k^2,x) + o(k^{m_0-1}), &
 E+\varepsilon>\la
\ge E,
\end{cases}
\eeq
where $f^-_1(z,x), f^-_2(z,x)$ are real-valued functions which are
polynomials of degree $m_0-1$ with respect to $z$ and differentiable with
respect to $x$. Next, define
\beq\label{defin6}
\breve\phi_-(\la,x,0) = \begin{cases}
f^-_1(h^2,x) - \I h f^-_2(h^2,x), & \la \le E,\\
f^-_1(-k^2,x) - k  f^-_2(-k^2,x), & \la \ge E,
\end{cases}
\eeq
and note that  $\ov{\phi_-(\la,x,0)} = \breve\phi_-(\la,x,0)
+o(h^{m_0-1})$ for $E-\varepsilon<\la \le E$.

Similarly, we  write
\beq\label{defin7}
\phi_+(\la,x,0) = \begin{cases}
f^+_1(h^2,x) + h f^+_2(h^2,x) + o(h^{m_0-1}), &E-\varepsilon< \la \le E,\\
f^+_1(-k^2,x) - \I k f^+_2(-k^2,x) + o(k^{m_0-1}), &
E+\varepsilon>\la \ge E,
\end{cases}
\eeq
and define
\beq\label{defin8}
\breve\phi_+(\la,x,0) = \begin{cases}
f^+_1(h^2,x) - h f^+_2(h^2,x) , & \la \le E,\\
f^+_1(-k^2,x) + \I kf^+_2(-k^2,x), & \la \ge E,
\end{cases}
\eeq
implying $\ov{\phi_+(\la,x,0)} = \breve\phi_+(\la,x,0)+o(k^{m_0-1})$
for $E+\varepsilon>\la \ge E$. In particular, note that
\[
\I^j \frac{\pa^j}{\pa h^j} \breve\phi_\pm(\la,x,0) =  \frac{\pa^j}{\pa k^j}
\breve\phi_\pm(\la,x,0), \qquad \la=E,\quad 0 \leq j \leq m_0-1.
\]
Moreover,
\beq \label{gipm}
g_\pm(\la,0) = \pm
\wronsk(\phi_\pm(\la,.,0),\breve\phi_\pm(\la,.,0)) +o\big((\la-E)^{(m_0-1)/2}
\big),\quad \la\in (E-\varepsilon,E+\varepsilon).
\eeq
While the above Wronskian depends on $x$ ($\breve\phi_\pm$ do not solve \eqref{S.4}
in general), the leading order is independent of $x$. Here and in all following Wronskians below, we set $x=0$
(and of course $t=0$).
Now consider  (cf. \eqref{defP})
\[
P(\la) = \frac{-\I}{2g_+(\la,0)g_-(\la,0)\wronsk_0(\phi_-,\phi_+)
\wronsk_0(\overline{\phi_-},\phi_+)}, \quad \la < E,
\]
and set
\beq
\tilde{P}(\la) = \frac{\I}{2}
\frac{\wronsk(\phi_-,\breve\phi_-)\wronsk(\phi_+,\breve\phi_+)}
{\wronsk(\phi_-,\phi_+)\wronsk(\breve\phi_-,\phi_+)},
\quad \la \in (E-\varepsilon,E+\varepsilon),
\eeq
implying
\[
P(\la) = \tilde{P}(\la)+o\big((\la-E)^{(m_0-2)/2} \big),
\quad E-\varepsilon<\la\leq E
\]
In particular,\footnote{Note that  if $\wronsk(\phi_+,\phi_-)(E)=0$, we loose one derivative, in which case
we have $\wronsk(\phi_+,\phi_-)=C k(1+o(1))$ by Lemma~\ref{lem2.3} {\bf II.~(b)}.}
\beq
\frac{\pa^j}{\pa h^j} P(\la) = \I^j \frac{\pa^j}{\pa k^j}
\tilde{P}(\la), \qquad \la=E,\quad 0 \leq j \leq m_0-2.
\eeq
Using the Pl\"ucker identity
\[
\wronsk(f_1,f_2) \wronsk(f_3,f_4) + \wronsk(f_1,f_3)
\wronsk(f_4,f_2) + \wronsk(f_1,f_4) \wronsk(f_2,f_3) =0
\]
with
$f_1=\phi_-$, $f_2=\breve\phi_-$, $f_3=\phi_+$, and
$f_4=\breve\phi_+$ we can rewrite $\tilde{P}(\la)$
 as
\[
\tilde{P}(\la) = \frac{\I}{2}\left(-\frac {V(k)}{W(k)}
 +\frac{\breve W(k)}{\breve V(k)} \right),
\]
where
\[
W(k):= \wronsk(\phi_-,\phi_+), \quad V(k):=
\wronsk(\phi_-,\breve\phi_+),
\]
and
\[
\breve W(k):=
\wronsk(\breve \phi_-,\breve\phi_+), \quad \breve V(k):=
\wronsk(\breve\phi_-,\phi_+).
\]
Moreover, using \eqref{defin5}--\eqref{defin8},  one can verify that
\[
\frac{\breve W(k)}{\breve V(k)} =
\ov{\left(\frac{V(-k)}{W(-k)}\right)} + o(k^{m_0-1}).
\]
Now, since  $V(k)= \wronsk_0(\phi_-,\ov{\phi_+}) + o(k^{m_0-1})$,
we obtain
\beq R(k) =
-\frac {V(k)}{W(k)} + o(k^{m_0-2}).
\eeq
This implies
\beq
\Re\left(\I^{j+1}\,\frac{\pa^j }{d k^j} R(0)\right)= \I^j
\frac{\pa^j}{\pa k^j} \tilde{P}(0) =
\frac{\pa^j P}{\pa h^j}(0), \qquad j=0,\dots,m_0-2,
\eeq
and we are done.
\end{proof}

Finally we discuss the possibility of integrating the last (unbounded)
integrand of \eqref{sumFc} by parts. More precisely, we discuss the boundary
terms corresponding to the point $E_{2r_++3}^+=+\infty$ (again the
considerations are the same for the $+$ and $-$ cases, and we study only
the $+$ case). To begin, we recall the well-known asymptotic expressions
\[
m_+(\la)=\I\sqrt{\la}(1+o(1)),\quad \theta(\la)=\sqrt{\la}(1+o(1))
\]
as $\la\to\infty$. Moreover, (recall $\alpha_+(\la,t)=4\I(\sqrt\la)^3t (1+o(1))$)
we also have
\[
\frac{\pa^s}{\pa\theta^s} u_+(\la,x,t)=O(1), \qquad
\frac{\pa^s}{\pa\theta^s} \E^{\alpha_+(\la,t)} =O(t \theta^{2s}).
\]
As in the previous cases, the only interesting part is  the reflection
coefficient $R_+(\la,0)$ for which we have
\beq
\frac{\pa^s}{\pa\theta^s} R_+(\la,0)= O(\theta^{-n_0-1}),
\qquad s=0,\dots, m_0,
\eeq
as $\la\to\infty$ by Lemma~\ref{lem2.3} {\bf III.~(a)}. Hence we conclude that
\[
\frac{\pa^s}{\pa\theta_+^s} \rho_+(\la(\theta),x,y,t) =O(\theta^{2s-n_0-1}),
\quad \mbox{as } \la\to\infty,
\]
uniformly with respect to $x,y\in\R$ and $t\in[0,T]$ for any $T>0$. As a consequence
we can perform $m \leq \floor{\frac{n_0}{2}}$ partial integrations such that the boundary
terms at $\infty$ vanish.

In summary, supposing \eqref{2.111}, the maximum number of
integration by parts is determined by Lemma~\ref{lemest4}; that is it is determined
by the points $\inte \si_+\cap\pa\si_-$, and is given by
$\floor{\frac{m_0}{2}}$. So, excluding case 3), we have an
almost complete picture. However, up to this point, we have only
looked at $F_{R,+}(x,y,t)$ and have not considered derivatives with
respect to $x,y,t$. Fortunately, since $R_+(\la,0)$ is evidently
independent of these variables and all other terms of the integrand
\eqref{sumFc} (except for the last summand
in \eqref{sumFc}) can be differentiated as often as we please, these
derivatives do not affect our analysis. Moreover, for this last summand, one has only to
take into account that a partial derivative with respect to $x$ or
$y$ adds $O(\theta_+)$ (from $\E^ {\I\theta_+(x+y)}$) and a partial
derivative with respect to $t$ adds $O(\theta_+^3)$ (from
$\E^{\alpha_+(\la,t)}$). Thus we obtain the following result.

\begin{lemma}\label{lemfinal}
Let $q(x)$ satisfy \eqref{2.111}. Then
\beq\label{est10}
F_+(x,y,t)=\frac{1}{(x+y)^{\floor{\frac{m_0}{2}}}}
\left(H(x,y,t) +\int_{\tilde E}^{+\infty} \E^{\I\sqrt\la(x+y)
+4\I(\sqrt\la)^3 t}\,H_1(\la,x,y,t)d\la\right),
\eeq
where $\tilde E> \max\{E_{2r_+}^+, E_{2r_-}^-, 1\}$.
The function $H(x,y,t)$ is smooth on the set  $\mathcal D:=
[0,+\infty)\times[0,+\infty)\times[0,T]$. All
partial derivatives with respect to $x,y,t$ of function $H$ are
bounded on $\mathcal D$.
The function $H_1(\la,x,y,t)$ is bounded on $\la$ and smooth with
respect to $x,y,t\in\mathcal D$.  Moreover,
\beq\label{main5}
\frac{\pa^{l+s+k}}{\pa x^l\pa y^s\pa
t^k}\,H_1(\la,x,y,t)=o\left(\left(\sqrt\la\right)^{l+s+3k-n_0-2
+2 \floor{\frac{m_0}{2}}}\right)\quad \mbox{as } \la\to\infty,
\eeq
uniformly on $\mathcal D$.
\end{lemma}

Lemma \ref{lemfinal} shows that for the  integral in \eqref{est10} and its
derivatives with respect to $x,y,t$  to converge, it is
sufficient that $l+s+3k +2\floor{\frac{m_0}{2}}-n_0<0$.
A comparison to \eqref{4.311} shows that to guarantee a
classical solution (three derivatives with respect to $x$ and
one with respect to $t$) of the KdV equation, we need  at least
$l+s= 4, k=0$ and  $l+s= 1, k=1$ to hold; that is, we need $m_0$ and $n_0$ to satisfy $4+2\floor{\frac{m_0}{2}}-n_0<0$.
Since we also need $\floor{\frac{m_0}{2}}-2\geq 2$, this yields the conditions
$m_0\geq 8$ and $n_0 \geq 2\floor{\frac{m_0}{2}} +5$.

In particular, if \eqref{2.111} holds for all $m_0, n_0\in \N$ (Schwartz-type
perturbations), then the same is true for the solution. Thus this provides
a generalization of the main result from \cite{EGT} without any restriction
on the background spectra.

\bigskip
\noindent{\bf Acknowledgments.} We are very grateful to F.\ Gesztesy, E.Ya.\ Khruslov,
and V.A.\ Marchenko for helpful discussions. I.E.\ gratefully acknowledges the
extraordinary hospitality  of the Faculty of Mathematics at the University of
Vienna during extended stays 2008--2009, where parts of this paper were
written. G.T.\ gratefully acknowledges the stimulating atmosphere at the
Centre for Advanced Study at the Norwegian Academy of Science and Letters in Oslo
during June 2009 where parts of this paper were written as part of  the international
research program on Nonlinear Partial Differential Equations.

\end{document}